\documentclass{llncs}
\usepackage{cite,microtype,graphicx}
\usepackage{hyperref,color}
\usepackage{amsmath,amssymb}
\usepackage[capitalize,nameinlink]{cleveref}

\spnewtheorem{observation}[lemma]{Observation}{\bfseries}{\itshape}

\title{Limitations on Realistic\\ Hyperbolic Graph Drawing}

\author{David Eppstein}
\authorrunning{D. Eppstein}
\institute{Computer Science Department\\
University of California, Irvine\\
Irvine, CA 92697, USA\\
\email{eppstein@uci.edu}}

\date{ }

\begin{document}
\maketitle  

\begin{abstract}
We show that several types of graph drawing in the hyperbolic plane require features of the drawing to be separated from each other by sub-constant distances, distances so small that they can be accurately approximated by Euclidean distance. Therefore, for these types of drawing, hyperbolic geometry provides no benefit over Euclidean graph drawing.

\end{abstract}

\section{Introduction}

Although most graph drawing algorithms place vertices and edges in the Euclidean plane, several past works instead use a different geometry, the hyperbolic plane. Beginning in the 1990s, researchers proposed hyperbolic graph drawings to combine focus and context: the fisheye-like view provided by the Poincar\'e disk visualization of the hyperbolic plane allows parts of the drawing to be shown in an expanded view, with the rest compressed into the margins of the Poincar\'e disk but remaining entirely visible~\cite{LamRao-JVLC-96}. This line of research also includes similar techniques using three-dimensional hyperbolic geometry~\cite{MunBur-VRML-95,Mun-InfoVis-97,Mun-CGA-98,TurBal-JCLC-14}. The hyperbolic plane has also been used for greedy graph drawings, with the property that a path to any vertex can be found by always moving to a neighboring vertex that is closer to the eventual destination. Unlike the Euclidean plane, the hyperbolic plane allows such drawings for any graph~\cite{Kle-INFOCOM-07,EppGoo-TC-11,BlaFriKat-JEA-20}. Hyperbolic geometry was central to our construction of Lombardi drawings for graphs of maximum degree three~\cite{Epp-DCG-14} and for Halin graphs~\cite{DunEppGoo-JGAA-12}. We have also developed algorithms for finding a good choice of initial views in hyperbolic visualizations~\cite{BerEpp-WADS-01}, and used spring embedding techniques to find high-quality hyperbolic graph drawings~\cite{KobWam-TVCG-05}. Other investigations of hyperbolic graph drawing include the use of circle packings to construct hyperbolic drawings~\cite{Moh-GD-99}, interactive systems using hyperbolic drawing~\cite{EklRobGre-CW-02,WalRit-KDD-02}, hyperbolic drawing of power-law graphs~\cite{BlaFriKro-TN-18}, hyperbolic Euler diagrams~\cite{SuzTakOno-ICML-19}, distance distortion of hyperbolic embeddings~\cite{Sar-GD-11,VerSur-CGTA-16,CouDuc-Eurocomb-17}, and hyperbolic multidimensional scaling~\cite{SalDeSGu-ICML-18,Wal-IS-04}.

In this paper, we investigate hyperbolic geometry from the point of view of \emph{realistic graph drawing}. This type of analysis, previously applied to Euclidean graph drawing~\cite{BarGooRil-JGAA-04,DunEfrKob-IJGCS-06}, treats the vertices and edges of a graph drawing as having nonzero radius or thickness, rather than being idealized mathematical points and curves, so that they can be seen by a reader of the drawing. This required thickness has been formulated mathematically in several related ways, including placing constraints on the vertex-vertex resolution (the minimum distance between center points of vertices) or vertex-edge resolution (the minimum distance of the center point of any vertex from an edge that it is not an endpoint of). In ink-based \emph{bold graph drawing} methods all features of the drawing must have visible parts that are not covered by other features, so that the graph may be unambiguously determined from its drawing~\cite{Kre-CGTA-11,Pac-JGAA-15}. These parameters are also related to \emph{angular resolution}, the sharpest angle between two edges incident at the same vertex, as edges forming sharp angles need high length to be visibly separated from each other~\cite{ForHagHar-SICOMP-93,GarTam-ESA-94,MalPap-SIDMA-94}.

Any Euclidean drawing can be scaled to achieve constant vertex-vertex or vertex-edge resolution, so these parameters are typically compared against the area of a bounding box of the drawing. A drawing style is considered to be good when it  achieves polynomial area, and bad when the area is exponential~\cite{DunEppGoo-DCG-13}. But in hyperbolic graph drawing, there is an absolute length scale, and it is not possible to rescale a drawing without changing its shape. This hyperbolic length scale is essential to focus+context applications of hyperbolic visualizations, as it controls the sizes of objects near the center of the visualization, relative to the overall view. In greedy drawings, constant vertex separation in this absolute length scale is necessary, because drawings with smaller vertex distances would be approximately Euclidean, constraining greedy drawings to have bounded vertex degree. More generally, parts of a hyperbolic graph drawing with features significantly smaller than the unit of absolute length would be approximately Euclidean, failing to take advantage of any differences between hyperbolic and Euclidean geometry. Therefore, we will define a realistic hyperbolic drawing to be one in which the resolution parameters of the realistic graph drawing model, such as vertex-vertex resolution, vertex-edge resolution, or line thickness, are at least constant in absolute length. We consider realistic drawing to be impossible when these parameters are forced by the constraints of the drawing to be $o(1)$.

Our work shows that this realistic model of hyperbolic graph drawing is severely limited, providing a partial explanation of the failure of hyperbolic approaches to  focus+context to come into wider use. In particular, we prove:

\begin{itemize}
\item Every straight-line crossing-free drawing of an $n$-vertex maximal planar graph in the hyperbolic plane has vertex-edge resolution $O(1/\sqrt{n})$ (\cref{thm:all-max-planar-have-small-ve-res}). Moreover, there exist $n$-vertex planar graphs (the well-known nested triangle graphs) for which every straight-line crossing-free drawing has vertex-edge resolution $O(1/n)$ (\cref{thm:nested-ve-res}). Both bounds are tight.

\item Although planar graphs have hyperbolic drawings with high vertex-vertex resolution, these drawings have exponentially small angular resolution for all maximal planar graphs (\cref{thm:max-planar-vv-angle}). This differs from Euclidean drawings, whose angular resolution is bounded by a function of the degree~\cite{MalPap-SIDMA-94}.

\item Simple structure is not enough to avoid these problems: some series-parallel graphs of bounded bandwidth require polynomially small vertex-edge resolution and either small vertex-vertex resolution or exponentially small angular resolution (\cref{thm:serpar}). Grid graphs also obey similar bounds (\cref{thm:grid}).

\item Beyond planar graph drawing, every $n$-vertex graph has a Euclidean drawing with unit vertex-vertex resolution and angular resolution $\Theta(1/n)$. However, we prove that for hyperbolic drawings with unit vertex-vertex resolution, some graphs require angular resolution $O(1/n^2)$ (\cref{thm:complete-angles}) and their bold drawings require edge width $O(1/n)$ (\cref{thm:complete-bold}).\end{itemize}

\section{Vertex-edge resolution}

\begin{figure}[t]
\centering\includegraphics[width=0.5\textwidth]{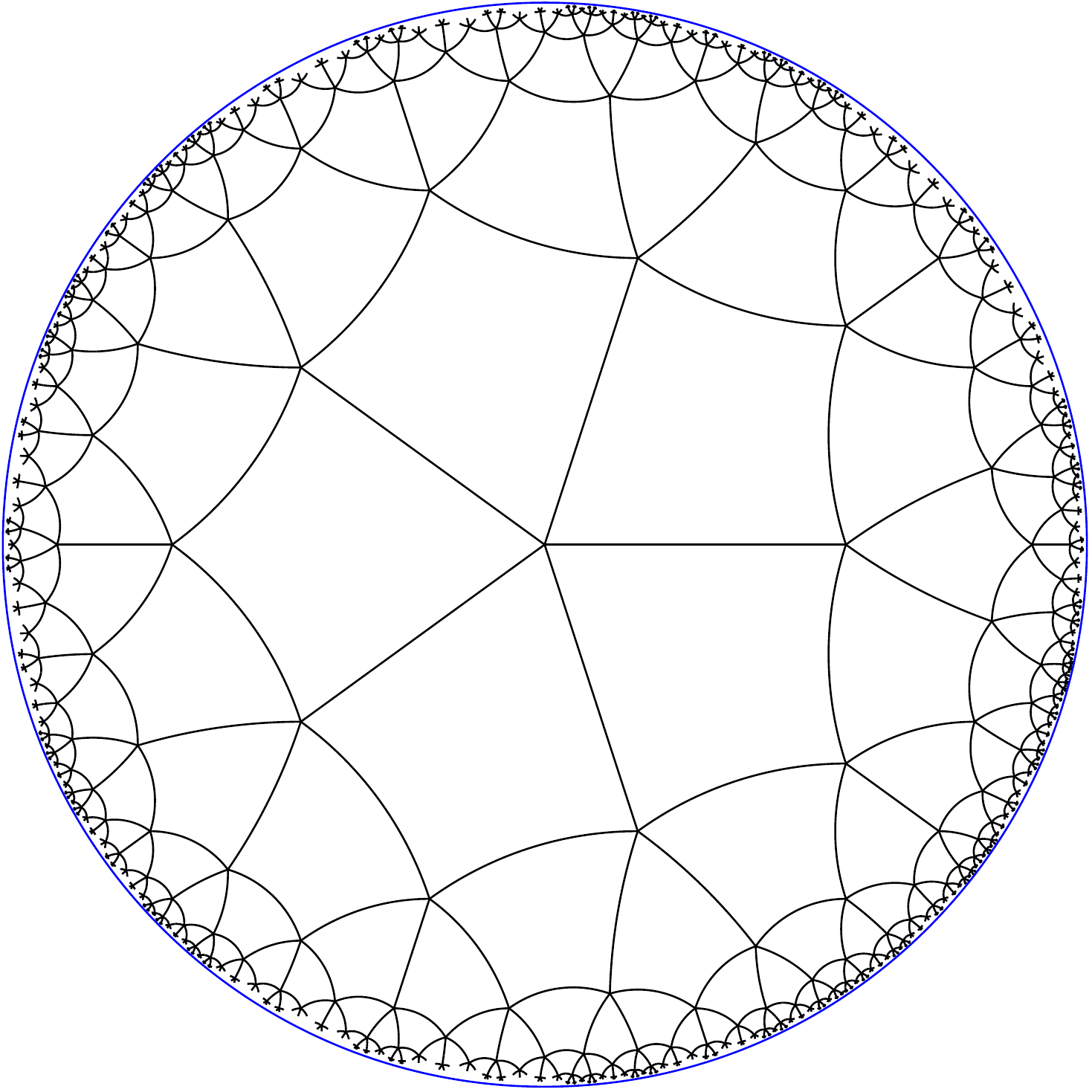}
\caption{A tessellation of the Poincar\'e disk model of the hyperbolic plane by squares. Arbitrarily large subgraphs of this graph have constant vertex-edge resolution, unlike the maximal planar graphs. Although the internal faces of this tessellation can be triangulated, it is impossible to add edges that make the outer face triangular without violating planarity.}
\label{fig:45tess}
\end{figure}

In this section, we examine the vertex-edge resolution (the minimum hyperbolic distance between a vertex and an unrelated edge) of graphs, drawn in the hyperbolic plane with (hyperbolically) straight edges. There exist planar graphs that can be drawn in this way with at least constant vertex-edge resolution (\cref{fig:45tess}). However, as we show, this is not the case for any maximal planar graph.

\begin{lemma}
\label{lem:max-triangle-area}
Every hyperbolic triangle has area at most $\pi$.
\end{lemma}

\begin{proof}
This follows from the well-known formula for the area of a hyperbolic triangle as $\pi-\sum\theta_i$, where $\theta_i$ are the internal angles of the triangle.
\end{proof}

\begin{lemma}
\label{lem:small-area-triangle}
In a planar straight-line hyperbolic drawing of an $n$-vertex maximal planar graph, at least one face has area $\le \frac{\pi}{2n-3}$.
\end{lemma}

\begin{proof}
By \cref{lem:max-triangle-area} the exterior face has area $\le\pi$. The remaining $2n-3$ faces partition this area into disjoint subsets, one of which must be $\le \pi/(2n-3)$.
\end{proof}

\begin{lemma}[tangent rule for hyperbolic right triangles]
\label{lem:tangent-rule}
If a hyperbolic right triangle has legs of length $x$ and $y$, the angle $\theta$ opposite $x$ satisfies
\[ \tan\theta=\frac{\tanh{x}}{\sinh{y}}. \]
\end{lemma}

\begin{proof}
See \cite[Corollary 32.13, p. 431]{Mar-82}.
\end{proof}

\begin{lemma}
\label{lem:right-area}
A hyperbolic right triangle with leg lengths $x\le y\le 1$ has area  $\Theta(xy)$.
\end{lemma}

\begin{proof}
This follows by expressing the area as $\pi$ minus the sum of angles, expressing these angles in terms of $x$ and $y$ according to \cref{lem:tangent-rule}, replacing these expressions by their power expansions, and omitting lower-order terms; see \cite[p. 434]{Mar-82}.
\end{proof}

\begin{lemma}
\label{lem:height-to-area}
Let $h$ be the height (minimum distance from any vertex to the opposite edge) of a hyperbolic triangle $T$. Then $T$ has area $\Omega\bigl(\min(1,h^2)\bigr)$.
\end{lemma}

\begin{proof}
Let height $h$ be achieved by a line segment from vertex $v$ to the opposite side $S$. In order avoid smaller height at another vertex, $S$ extends for distance at least $h/2$ on either side of this segment. The result follows by applying \autoref{lem:right-area} to the right triangles formed by the segment of length $h$ and the perpendicular segments of $S$ of length $h/2$. They are disjoint and lie entirely within $T$, so their total area lower-bounds that of $T$.
\end{proof}

\begin{theorem}
\label{thm:all-max-planar-have-small-ve-res}
Every straight-line planar hyperbolic drawing of an $n$-vertex maximal planar graph has vertex-edge resolution $O(1/\sqrt{n})$.
\end{theorem}

\begin{proof}
By \cref{lem:small-area-triangle}, at least one face has area $O(1/n)$. By \cref{lem:height-to-area}, the height of this face is $O(1/\sqrt{n})$. Therefore, this face has a vertex and non-incident edge that are at distance $O(1/\sqrt{n})$ from each other.
\end{proof}

\cref{thm:all-max-planar-have-small-ve-res} is tight: some $n$-vertex maximal planar graphs have vertex-edge resolution $O(1/\sqrt n)$. For instance, obtain $G$ from a square grid graph by triangulating each square and adding three surrounding vertices to form an outer face, shrink the drawing to have unit radius, and draw it within a unit disk of the Klein model of the hyperbolic plane (which preserves straight line drawings) giving a hyperbolic drawing whose vertex-edge resolution is the  scale factor, $O(1/\sqrt n)$.

\begin{figure}[t]
\centering\includegraphics[width=0.4\textwidth]{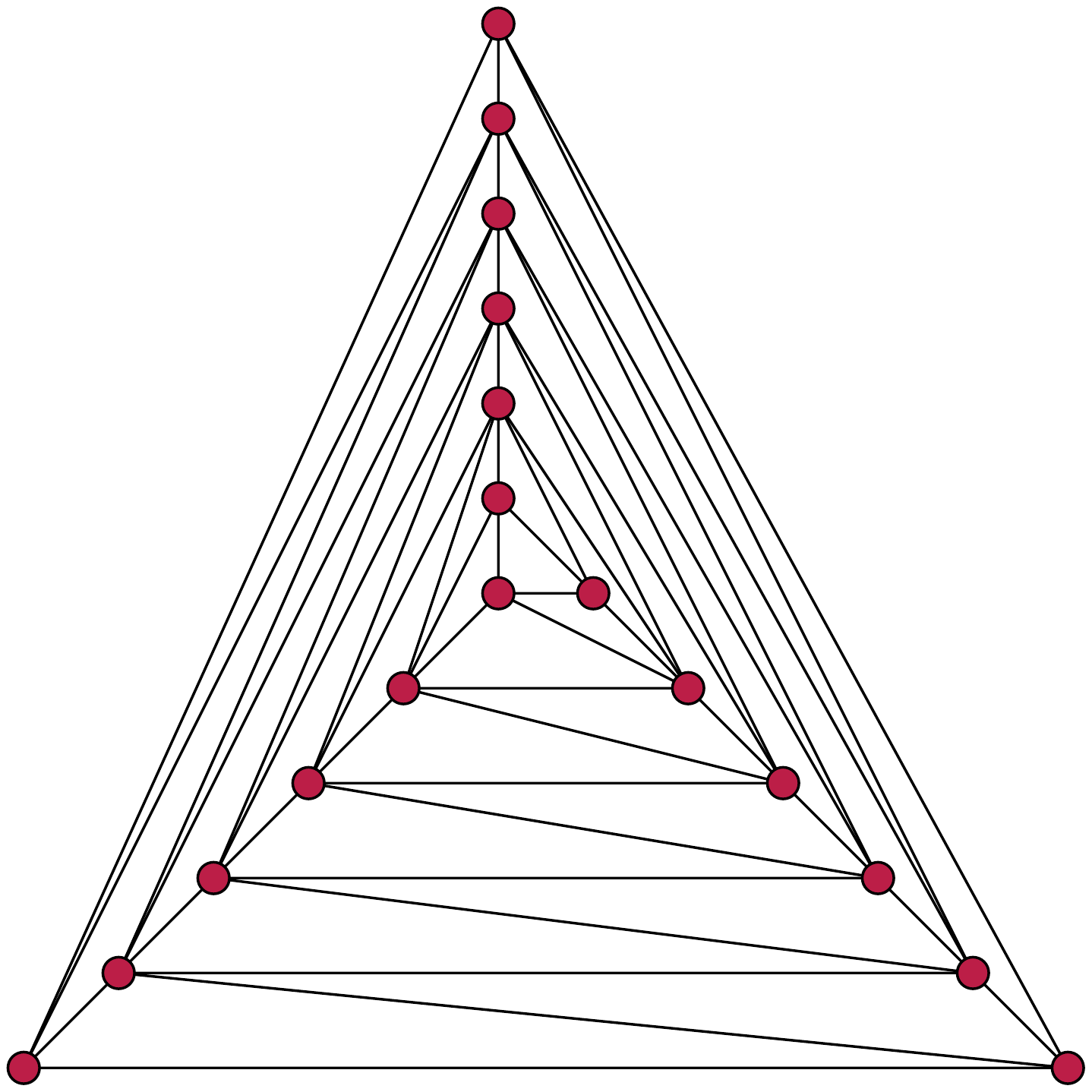}
\caption{A maximal planar version of the nested triangles graph, for which any planar straight-line hyperbolic drawing has vertex-edge resolution $O(1/n)$.}
\label{fig:nested}
\end{figure}

To strengthen \cref{thm:all-max-planar-have-small-ve-res} for some graphs, we use a maximal planar version of the \emph{nested triangles graph} (\cref{fig:nested}), formed from $n/3$ nested triangles by adding edges between consecutive triangles to make the graph maximal planar.

\begin{theorem}
\label{thm:nested-ve-res}
If the maximal planar nested triangles graph on $n$ vertices is given a straight-line planar drawing in the hyperbolic plane, then the drawing must have vertex-edge resolution $O(1/n)$.
\end{theorem}

\begin{proof}
Among the $n/3$ nested triangles of the graph, at least $n/6$ must be drawn as nested regardless of the choice of outer face. In any drawing, consider the middle triangle $T$ of these $n/6$ triangles. We distinguish two cases.
\begin{itemize}
\item If the hyperbolic diameter of triangle $T$ is $\ge 1$, then each of the $n/12$ rings of six faces that surround it in the drawn-nested subset of $n/6$ triangles must, to completely surround $T$, include at least two triangular faces of diameter $\Omega(1)$. Because these $n/6$ high-diameter triangular faces are all disjoint, and all lie within the $\le\pi$ area of the outer face of the drawing (\cref{lem:max-triangle-area}), one of them must have area $\le6\pi/n$. By \cref{lem:right-area}, this triangular face, with diameter $\Omega(1)$ and area $O(1/n)$, must have height $O(1/n)$.
\item If the hyperbolic diameter of triangle $T$ is $\le 1$, then $T$ surrounds a drawing of a nested triangles graph consisting of $n/4$ vertices in $n/12$ nested triangles, all drawn within a region of the hyperbolic plane of diameter $\le 1$. Within this region, hyperbolic distances can be approximated to within a constant factor by Euclidean distances. It is known that Euclidean straight-line drawings of the nested triangles graph must have two vertices whose distance is $O(1/n)$ times the diameter~\cite{DolLeiTri-ACR-84}, and the same follows for the hyperbolic drawing within $T$. If these two close-together vertices are adjacent on a single face of the drawing, then that face must have height $O(1/n)$, and otherwise they are separated by an edge and their distance from that edge is $O(1/n)$.
\end{itemize}
As both cases give a vertex-edge pair at distance $O(1/n)$, the result follows.
\end{proof}

Again, this result is tight. One may draw any planar graph in the hyperbolic plane with vertex-edge resolution $\Omega(1/n)$, in an essentially non-hyperbolic way, by first constructing a straight-line drawing in a Euclidean grid of size $O(n)\times O(n)$~\cite{FraPacPol-IJC-1990,Sch-SODA-1990,Bra-TGGT-08} and then using the Klein model of hyperbolic geometry to map this drawing to a straight-line grid drawing within a subset of the hyperbolic plane of diameter $O(1)$, with constant distortion of distances.

\section{Vertex-vertex resolution and angular resolution}

It is not possible to upper-bound only the vertex-vertex resolution of hyperbolic drawings of planar graphs, as the following observation shows.

\begin{observation}
For every planar graph $G$, and every distance $d$, there is a planar hyperbolic drawing of $G$ with vertex-vertex resolution $\ge d$.
\end{observation}

\begin{proof}
We may assume without loss of generality that $G$ is maximal planar, and use a de Fraysseix--Pach--Pollack~\cite{FraPacPol-IJC-1990} style graph drawing algorithm (without horizontal shifting)  in which vertices are placed into the drawing one-by-one in a canonical ordering, starting from two adjacent vertices on the outer face of the eventual drawing, so that when each vertex is added to the drawing it is adjacent to a consecutive subsequence of vertices on the outer face of the current drawing.

In the upper halfplane model of the hyperbolic plane, place the first two vertices on two arbitrary points with distinct $x$-coordinates. Place each subsequent vertex above the midpoint of the $x$-interval spanned by its earlier neighbors, so that at each stage the upper boundary of the drawing is an $x$-monotone piecewise-linear curve. Choose the vertical position of each vertex, both high enough that it is visible (along a hyperbolic line segment) to all of its previously-placed neighbors, and high enough that it is at distance at least $d$ from all previously placed vertices; these two requirements do not interfere with each other.
\end{proof}

Nevertheless, when the vertex-vertex resolution is large, it may force the drawing to be bad in other ways, as the remainder of this section shows.

\begin{figure}[t]
\centering\includegraphics[width=0.6\textwidth]{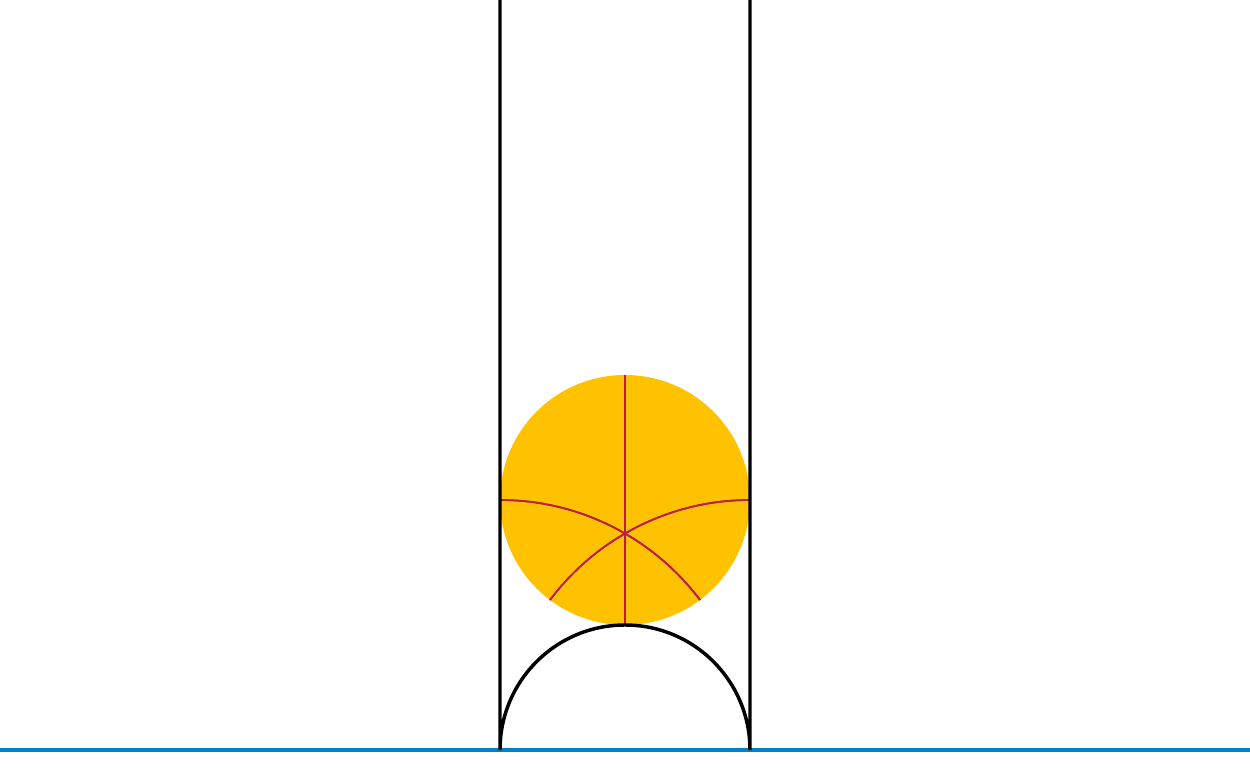}
\caption{An ideal hyperbolic triangle (black) in the upper halfplane model of the hyperbolic plane (above the blue line), with an inscribed circle of radius $\ln\sqrt{3}$ (yellow). The thin red hyperbolic line segments meet at the center of the circle (the incenter of the triangle), which is somewhat below its Euclidean center.}
\label{fig:ideal-incircle}
\end{figure}

\begin{lemma}
\label{lem:inradius}
Every hyperbolic triangle has an inscribed circle touching all three of its sides. Its radius (the inradius of the triangle) is at most $\ln\sqrt{3}\approx 0.5493$.
\end{lemma}

\begin{proof}
See for instance \cite[Theorem 13.4, p. 103]{Pet-20}. The limiting case of inradius${}=\ln\sqrt{3}$ occurs for \emph{ideal triangles} (\cref{fig:ideal-incircle}), with all vertices at infinity.
\end{proof}

\begin{lemma}
\label{lem:far-angle}
Let $v$ be a vertex of a hyperbolic triangle $T$, at distance $d$ from the incenter of $T$. Then the angle of $T$ at $v$ is upper bounded by an exponentially small function of $d$.
\end{lemma}

\begin{proof}
Assume without loss of generality that $d>2$ because otherwise the upper bound of the lemma is $O(1)$, trivially valid for all angles.

Let $V$ be the subset of $T$ between the inscribed circle $C$ and vertex $v$. Replace the sides of $T$ by two asymptotic lines, tangent to the inscribed circle of $T$ at points wider than the tangent points of $T$ to the same circle, and expand $C$ if necessary until the radius of the expanded circle $\hat C$ is exactly $\ln\sqrt{3}$, enclosing $V$ in the corresponding subset $\hat V$ of an ideal triangle $\hat T$ having the same incenter, with $v$ equally far from the two sides of $\hat T$. By choosing the point at infinity where the two sides of $\hat T$ meet to be the point at vertical infinity of an upper halfplane model of the hyperbolic plane, $\hat V$ can be made to be the region two vertical lines and above a circle $\hat C$ seen in the upper center of \cref{fig:ideal-incircle}.

In the upper halfplane model, the infinitesimal unit of length $ds$ is given by
\[ (ds)^2=\frac{(dx)^2+(dy)^2}{y^2}, \]
where $x$ and $y$ are Cartesian coordinates. This distance is locally Euclidean, with a scale factor inversely proportional to height, so two points on the vertical lines of \cref{fig:ideal-incircle} at Euclidean distance $y$ above the blue line are at hyperbolic distance $O(1/y)$ from each other. Integrating this scaled distance over a vertical line segment, the distance from $(x,y_1)$ to $(x,y_2)$ (with $y_1<y_2$) simplifies to $\ln(y_2/y_1)$. Therefore, for point $v$ to be at hyperbolic distance $d$ above the center of $\hat C$ in this model, its $y$-coordinate is exponentially larger than the $y$-coordinate of the center of $C$. At that height, the scale factor of hyperbolic distance is so small that the two vertical sides of $\hat T$ are exponentially close. The points along the two sides of $T$ at unit distance from $v$ are at a height corresponding to hyperbolic distance $\ge d-1$ above the center of $\hat C$, and are also exponentially close in hyperbolic distance, because they are sandwiched between the two sides of $\hat T$ that are exponentially close at that height.

The angle of $T$ at $v$ equals the angle formed at $v$ by these two exponentially-close points at unit distance from $v$, so it is exponentially small.
\end{proof}

\begin{theorem}
\label{thm:max-planar-vv-angle}
For every constant $c$, and every $n$-vertex maximal planar graph $G$, every hyperbolic planar straight line drawing of $G$ that has vertex-vertex resolution $\ge c$ also has angular resolution that is exponentially small in $n$, with the base of the exponential depending on $c$.
\end{theorem}

\begin{proof}
Let $T$ be the outer face of a drawing of $G$, and partition $T$ into four regions: its inscribed circle $C$, and the three regions between this circle and the three vertices of $T$. The inscribed circle has bounded diameter by \cref{lem:inradius} and can contain only $O(1)$ points of minimum separation $c$, so there must be a vertex $v$ of the outer face and a region $V$ between $C$ and $v$ that contains $\Omega(n)$ vertices of the drawing. Partition $V$ into regions of bounded diameter by classifying the points of $V$ according to their distance from the center of $C$, rounded to an integer; each of these regions can contain only $O(1)$ points of minimum separation $c$, so there must be a vertex $w$ within a region at distance $\Omega(d)$ from the center of $C$. Vertex $v$ itself must be even farther away, in order for $T$ to enclose $w$. Therefore, by \cref{lem:far-angle}, the angle of $T$ at $v$ is exponentially small.
\end{proof}

\section{Special classes of planar graphs}
\subsection{Series-parallel graphs}

The proofs of \cref{thm:all-max-planar-have-small-ve-res} and of \cref{thm:max-planar-vv-angle} depend only on the property of maximal planar graphs that they have a triangle (the outer triangle of their planar embedding) that contains a linear number of other vertices. Therefore, it can be adapted to other graphs with analogous properties.

\begin{figure}[t]
\centering\includegraphics[width=0.8\textwidth]{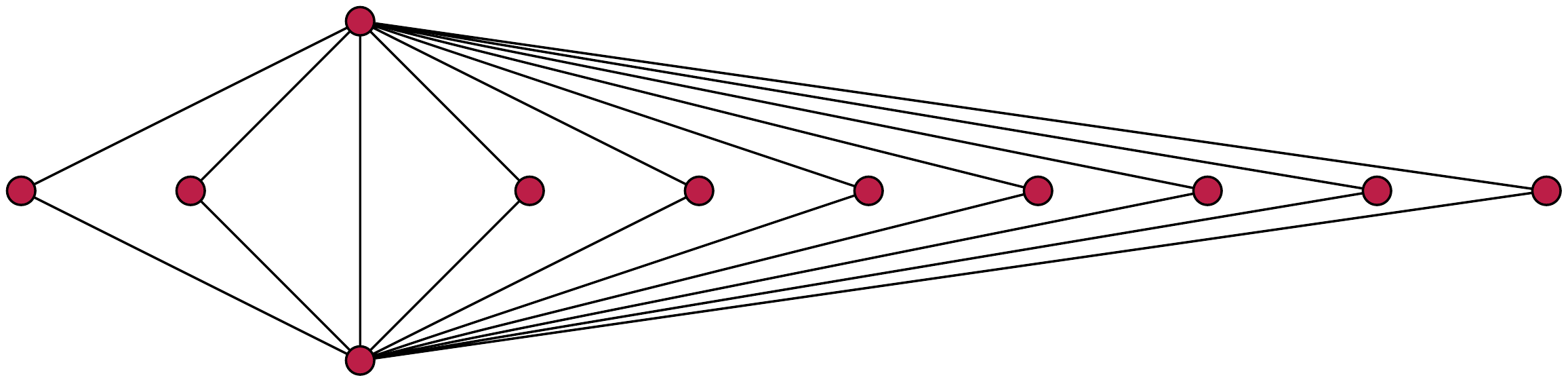}
\caption{$K_{1,1,9}$}
\label{fig:k119}
\end{figure}

\begin{theorem}
\label{thm:serpar}
For every $n$, there exists an $n$-vertex series-parallel graph $G$ of bounded bandwidth with the following properties:
\begin{itemize}
\item Every hyperbolic planar straight line drawing of $G$ has vertex-edge resolution $O(1/\sqrt{n})$.
\item For every $c$, every hyperbolic planar straight line drawing of $G$ that has vertex-vertex resolution~$\ge c$ also has angular resolution that is exponentially small in $n$, with the base of the exponential depending on $c$.
\end{itemize}
\end{theorem}

\begin{proof}
Let $G=K_{1,1,n-2}$ (\cref{fig:k119}). Up to permutation of the vertices all its planar straight-line drawings consist of $a$ nested triangles on one side of the edge between the two singleton sets of the tripartition, for $0\le a\le n-2$, and of $n-2-a$ nested triangles on the other side. Therefore, if $T$ is the outermost triangle on the side with the larger number of triangles, $T$ contains $\Omega(n)$ other vertices of $G$.

The bound on vertex-edge resolution follows the same lines as the proof of \cref{thm:all-max-planar-have-small-ve-res}: $T$ is partitioned into $\Omega(n)$ faces (one triangle and the rest quadrilaterals), so one of these faces has area $O(1/n)$.
If this low-area face is a triangle, the bounds on triangle height used in the proof of \cref{thm:all-max-planar-have-small-ve-res} show that it already has vertex-edge resolution $O(1/\sqrt{n})$. If it is a quadrilateral, add a diagonal, subdividing it into two triangles of area $O(1/n)$ and height $O(1/\sqrt{n})$. If the heights of both triangles are defined by the distances of a vertex to the added diagonal, then each of these vertices also has distance $O(1/\sqrt{n})$ to a non-adjacent side of the quadrilateral. If one or both of the triangles has a height that does not involve the added diagonal, then we get vertex-edge resolution $O(1/\sqrt{n})$ directly.

The result on vertex-vertex resolution and angular resolution follows directly by applying the argument from the proof of \cref{thm:max-planar-vv-angle} to~$T$.
\end{proof}

\subsection{Grid graphs}

The $n\times n$ grid graphs, in their standard Euclidean drawing, are particularly well-behaved: all edges have the same length, and the vertex-vertex, vertex-edge, and angular resolutions are all constant. Unlike the maximal planar graphs, they do not have cycles of bounded length containing many vertices. This does not prevent problems when drawing them hyperbolically:

\begin{theorem}
\label{thm:grid}
Every hyperbolic planar straight line drawing of an $n\times n$ grid has vertex-edge resolution $O(1/\sqrt{n})$. All such drawings that have vertex-vertex resolution $\Omega(1)$ have angular resolution exponentially small in~$n$.
\end{theorem}

\begin{proof}
If the grid is drawn with the standard outer face, a polygon with $4(n-1)$ sides, this polygon can be triangulated into $O(n)$ triangles. By \cref{lem:max-triangle-area} the bounded faces of the drawing cover an area of $O(n)$, and one of the $(n-1)^2$ grid quadrilaterals has area $O(1/n)$. For a different outer face, the area is even smaller. The argument from a quadrilateral of small area to small vertex-edge resolution is the same as for \cref{thm:serpar}.

Now suppose we have a drawing with vertex-vertex resolution $\Omega(1)$, and triangulate its boundary polygon. Let $v$ be any non-boundary vertex of the grid. Define two polygonal curves to the boundary from $v$, that step to the nearest points on two different sides of the triangle containing $v$ in the boundary triangulation, then repeatedly cross successive triangles on shortest crossing segments until reaching the boundary. The dual graph of the boundary triangulation (like any triangulation of any simple polygon) is a tree, and these curves follow a path in this dual tree, so they must terminate at the boundary, crossing each triangle at most once. We say that $v$ is captured by triangle $T$ if one of these curves crosses $T$, $v$ is within distance $\le \gamma n$ of the incenter of $T$ (for a constant of proportionality $\gamma$ to be determined later) or within unit distance of a vertex of $T$, and $v$ is not captured by any triangle crossed earlier by the curve.

If vertex $v$ is captured by triangle $T$, then before the curve for $v$ reaches $T$ it can only cross segments of other triangles that are exponentially short (as a function of $n$), by the same reasoning as in \cref{lem:far-angle}. Therefore, the grid vertices captured by $T$ are exponentially close to $T$. Partitioning the region close to $T$ into subsets of bounded diameter in the same way as in the proof of \autoref{thm:max-planar-vv-angle} shows that, if $T$ captures $k_T$ grid vertices, then one must be at distance $\Omega(k)$ from the incenter and more than unit distance from all vertices of $T$. However, captured vertices are defined as being within distance $\le \gamma n$ of the incenter or unit distance of a vertex of $T$, so $k_T=O(\gamma n)$. The boundary triangulation has $O(n)$ triangles, so the total number of captured vertices over all triangles is $O(\gamma n^2)$. If we choose $\gamma$ to be a sufficiently small constant (depending on the vertex-vertex resolution), this number will be less than the number $(n-2)^2$ of interior vertices of the grid, and at least one interior vertex $v$ will remain uncaptured.

Each of the $O(n)$ steps in the two curves from $v$ to the boundary has exponentially-small length, so both paths are exponentially short. Therefore, $v$ is sandwiched between two exponentially-close edges of the boundary, and moreover is at least unit distance from the edge endpoints. Because $v$ is an interior vertex of the grid, it has degree four, and has at least two edges extending in at least one of the two directions approximately parallel to the boundary edges that sandwich~$v$. These two edges must form an exponentially small angle at~$v$.
\end{proof}

\section{Nonplanar drawings}

Beyond planar graph drawing, any $n$-vertex graph has a Euclidean drawing with unit vertex-vertex resolution and angular resolution $\Theta(1/n)$, with vertices on a unit regular $n$-gon. The analogous placement in hyperbolic geometry has angular resolution $\Theta(1/n^2)$. In this section, we prove that this is optimal: every hyperbolic drawing of $K_n$ with unit vertex-vertex resolution has angular resolution $O(1/n^2)$.

A key ingredient is the relation between hyperbolic circle area and perimeter:

\begin{lemma}
\label{lem:perim-area}
Every hyperbolic circle of perimeter $p>1$ has area $p+o(p)$.
\end{lemma}

\begin{proof}
This follows immediately from the formulas for the area $4\pi\sinh^2(r/2)$
and perimeter $2\pi\sinh r$ of a hyperbolic circle of radius~$r$~\cite{Ben-JOMA-01}. In the limit as $r$ becomes large, the ratio of these two formulas converges to~$1$.
\end{proof}

We also need the following counterintuitive property of hyperbolic geometry, which in this respect is very different from Euclidean geometry:

\begin{lemma}
\label{lem:quadrants}
Divide the hyperbolic plane into four quadrants by two perpendicular lines. Then every line from a point in one quadrant to a point in the opposite quadrant passes within distance $\ln(1+\sqrt{2})\approx 0.8814$ of the crossing.
\end{lemma}

\begin{proof}
Choose an upper halfplane model of the hyperbolic plane that represents the two perpendicular lines as congruent semicircles, and the two opposite quadrants crossed by any given line as the two congruent regions to the left and right of their crossing (\cref{fig:double-wedge-hull}). The convex hull of the two quadrants is bounded by two more hyperbolic lines, asymptotic to the crossing lines (the boundaries of the light yellow region in the figure). All lines from one quadrant to another remain within the hull, crossing the red circle in the figure.

The figure may be given Cartesian coordinates in which the semicircles representing the two perpendicular lines have (Euclidean) centers at the points $(-1,0)$ and $(1,0)$ and cross at the point $(0,1)$, the hyperbolic center of the red circle. With these coordinates, these semicircles have Euclidean radius $\sqrt{2}$. The top point of the red circle is at $(0,1+\sqrt{2})$, and the result follows from the formula $\ln(y_2/y_1)$ for the hyperbolic length of a vertical line segment.
\end{proof}

\begin{figure}[t]
\centering\includegraphics[width=0.5\textwidth]{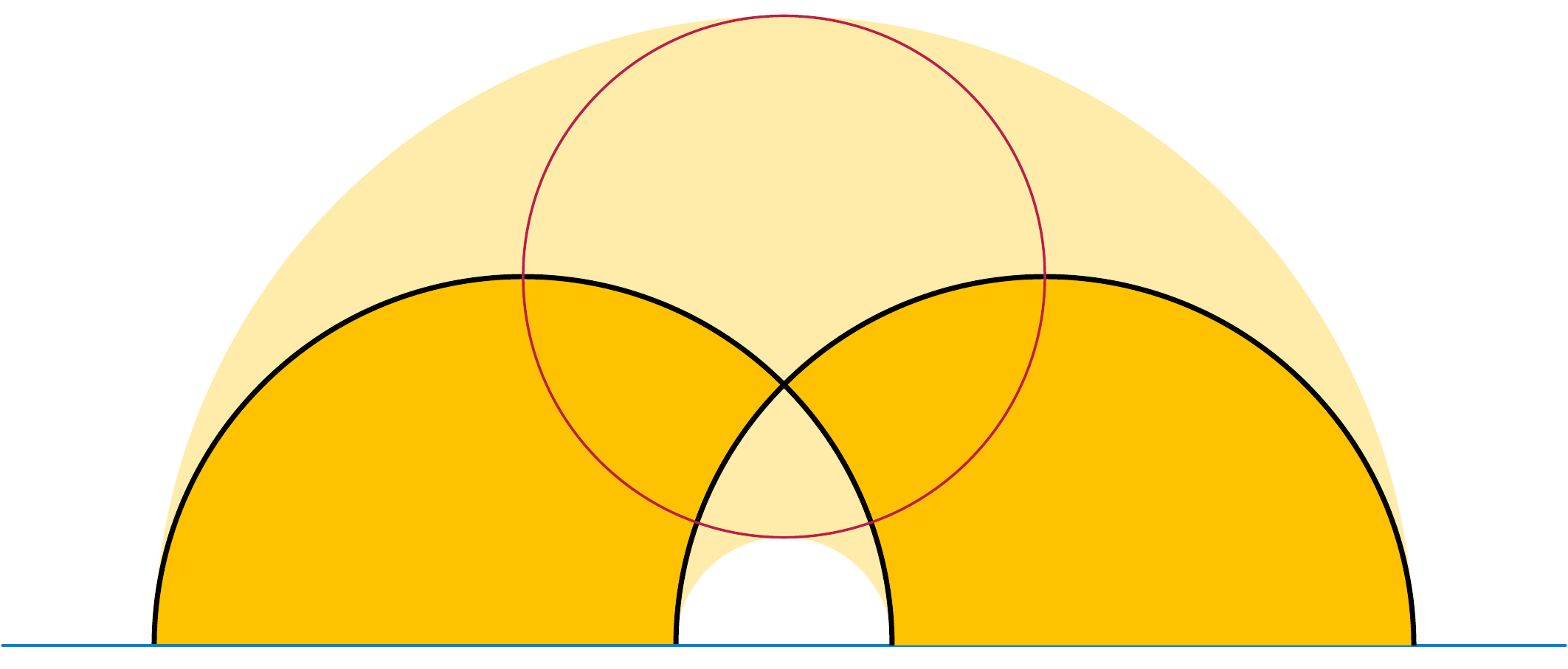}
\caption{Illustration for \cref{lem:quadrants}. Two perpendicular lines (black) in the upper halfplane model (above the blue line) determine two opposite quadrants (dark yellow) whose convex hull (light yellow) is bounded by two lines asymptotic to the two perpendicular lines. All lines from the left quadrant to the right quadrant remain within the hull and pass through the red circle, of radius $\ln(1+\sqrt{2})$.}
\label{fig:double-wedge-hull}
\end{figure}

\begin{theorem}
\label{thm:complete-angles}
For every constant $c$, every hyperbolic drawing of the complete graph $K_n$ with vertex-vertex resolution $\ge c$ has angular resolution $O(1/n^2)$.
\end{theorem}

\begin{proof}
By ignoring one vertex if necessary, assume without loss of generality that $n$ is even. Consider any drawing with vertex-vertex resolution $\ge c$, find a hyperbolic line splitting the vertices of the drawing into equal subsets, and (by applying the intermediate value theorem to the partitions by perpendicular lines) find a second perpendicular line splitting the vertices into two equal subsets. Let $x$ be the crossing point of these two lines. Among the four quadrants formed by these two lines, opposite quadrants necessarily have equal numbers of points, so two opposite quadrants $Q$ and $Q'$ each contain at least $n/4$ vertices.

Within $Q$, because of the vertex spacing, not all vertices can be in a quarter-circle centered at $r$ of area $o(n)$. Therefore, by \cref{lem:perim-area}, for some vertex $v$ in $Q$, a circle of radius $xv$ has perimeter $\Omega(n)$. If we center such a circle at $v$, then the circle of radius $\ln(1+\sqrt{2})$ centered at $x$ (the red circle in \cref{fig:double-wedge-hull}) spans only a constant number of units of the perimeter, an $O(1/n)$ fraction of the total perimeter, so the angle spanned by the red circle as viewed from $v$ is $O(1/n)$.

Within this $O(1/n)$ angle as viewed from $v$, at least $n/4$ vertices in $Q'$ are also visible, by \cref{lem:quadrants}. By the pigeonhole principle, some two of these vertices in $Q'$ must be within an angle of $O(1/n^2)$ of each other as viewed from $v$.
\end{proof}

Following van Kreveld~\cite{Kre-CGTA-11} we define a \emph{bold} hyperbolic drawing to draw vertices as hyperbolic disks of a given radius and edges as thickened hyperbolic line segments of width less than this radius, with all edges and all vertices having part of their boundary visible.

\begin{theorem}
\label{thm:complete-bold}
Any bold hyperbolic drawing of $K_n$ has edge width $O(1/n)$.
\end{theorem}

\begin{proof}
Let the edge width of a drawing be $w<1$. As in the proof of \cref{thm:complete-angles}, find two perpendicular lines partitioning the centers of the vertices into quadrants, with two opposite quadrants $Q$ and $Q'$ containing at least $n/4$ vertex centers each; let $C$ be the circle of radius $\ln 1+\sqrt{2}$ centered at the crossing point, and let $D$ be the diameter of $C$ halfway between $Q$ and $Q'$. Each vertex in $Q$ has $\ge n/4$ edges to $Q'$, all crossing $D$, among which $O(1/w)$ can have an exposed portion of boundary between the vertex and $D$, because each edge whose boundary is not entirely covered by edges from the same vertex covers a segment of $D$ of length at least $w$. Similarly, each vertex in $Q'$ has $O(1/w)$ edges with exposed boundaries. Unless $w$ is $O(1/n)$, the total number of edges with exposed boundaries either near their endpoint in $Q$ or near their endpoint in $Q'$ will be less than the $(n/4)^2$ number of edges from $Q$ to $Q'$, and at least one edge will be totally covered.
\end{proof}

\section{Conclusions}
We have performed an initial investigation into hyperbolic graph drawing under realistic graph drawing models, showing that for many variations of these models, drawings are impossible or seriously limited. Other questions in this area, which we leave open for future research, include:

\begin{itemize}
\item Which planar graphs have planar hyperbolic drawings with bounded vertex-edge resolution? These include all trees, and all outerplanar graphs (using a placement of vertices on a large regular polygon), but not all planar graphs and not even all bounded-bandwidth series-parallel graphs. Are there other natural classes of planar graphs that always have such drawings?

\item For Euclidean planar drawings, edge-edge resolution is not usually studied separately, because it is essentially the same as vertex-edge resolution. However, in the hyperbolic plane, edges may approach each other closely even when all vertex-edge pairs are well separated. Does this cause differences between hyperbolic vertex-edge resolution and edge-edge resolution?

\item RAC graphs (graphs drawn with right-angle crossings) are motivated by realistic graph drawing: crossings with high angles are easier to understand than sharp crossing angles~\cite{DidEadLio-TCS-11}. However their definition strongly depends on geometry: edges that cross at right angles in the Euclidean plane have bipartite intersection graphs such as cycles of four edges. In the hyperbolic plane, odd cycles of right-angle-crossing edges are possible; however, 4-cycles are not possible. How do these differences affect the hyperbolic RAC graphs?
\end{itemize}

Our results should not be interpreted as shutting off research on hyperbolic graph drawing, which remains important in applications such as greedy routing where realistic drawing assumptions do not fit the problem, and as a building block for Euclidean drawing methods such as Lombardi drawing.

\bibliographystyle{splncs04}
\bibliography{hyperbolic}
\end{document}